\documentclass[12pt, reqno]{amsart}

\usepackage{amsmath,amsfonts,amsthm,amssymb,amsxtra, amsbsy, dsfont,bm,color}
\usepackage{color}
\usepackage{hyperref}
\usepackage{cite}
\usepackage{mathtools}

\newtheorem{thm}{Theorem}[section]
\newtheorem{cor}[thm]{Corollary}
\newtheorem{lem}[thm]{Lemma}

\theoremstyle{definition}

\newtheorem{rem}[thm]{Remark}

\newtheorem{prop}[thm]{Proposition}

\renewcommand{\epsilon}{\varepsilon}
\newcommand{\E}{\mathcal{E}}
\newcommand{\Eh}{\widehat{\mathcal{E}}}
\newcommand{\eh}{\widehat{E}}

\renewcommand{\phi}{\varphi}
\newcommand{\R}{\mathbb{R}}
\renewcommand{\S}{\mathbb{S}}

\newcommand{\be}{\begin{equation}}

\newcommand{\uR}{\underline{R}}
\newcommand{\uZ}{\underline{Z}}
\newcommand{\1}{{\ensuremath{\mathds{1}} }}

\def\tr{\mathop{\mathrm{tr}}\nolimits} % Spur

\begin{document}

\title{Binding of atoms in M\"uller theory}
% Force line breaks with \\

\author{Yukimi Goto}

\address{Department of Mathematics, Gakushuin University 
1-5-1 Mejiro Toshima-ku Tokyo 171-0031, Japan.}
\email{\tt yukimi@math.gakushuin.ac.jp}

\begin{abstract}
%{\bf Abstract.}
We give a necessary and sufficient condition for the existence of molecules in  M\"uller theory.
Furthermore, we show that if a system is stable in Born-Oppenheimer approximation, then the bound on the positive excess charge $ Z-N \le cZ^{1-\epsilon}$ follows.
\end{abstract}

\maketitle

%%%%%%%%%%%%%%%%%%%%%%%%%%%%%%%%%%%%%%

\section{Introduction}
We consider a molecule with $N>0$ electrons and $K$ nuclei.
We say that a self-adjoint operator $\gamma$ is an one-body density-matrix if $0 \le \gamma \le 1$ on $L^2(\R^3)$ and $\tr \gamma < +\infty$.
Then the M\"uller functional is defined by
\[
\mathcal{E}_{\uR}(\gamma)
= \tr\left[ \left( -\frac{1}{2} \Delta - V_{\uR} \right) \gamma \right]
+D[\rho_{\gamma}] -X(\gamma^{1/2}),
\]
where $D[\rho_\gamma]$ is the direct part of Coulomb energy defined by
\[
D[\rho_{\gamma}] = D(\rho_{\gamma}, \rho_{\gamma}) = \frac{1}{2}\iint_{\R^3\times \R^3}
\frac{\rho_{\gamma}(x) \rho_{\gamma}(y)}{|x-y|}dxdy
\]
and the M\"uller exchange energy is defined by
\[
X(\gamma^{1/2}) = \frac{1}{2}\iint_{\R^3\times \R^3} \frac{|\gamma^{1/2}(x, y)|^2}{|x-y|}dxdy.
\]
Here $\gamma^{1/2}(x, y) = \sum_{i\ge1} \lambda_i^{1/2} \phi_i(x) \phi^*_i(y)$, with $\gamma \phi_i = \lambda_i \phi_i$, and $\rho_{\gamma} (x)= \gamma(x, x)$ is the one-particle electron density.
Our potential is
\[ V_{\uR}(x) = \sum_{i=1}^K \frac{Z_i}{|x-R_i|}, \quad Z=\sum_{i=1}^K Z_i,
\]
where $\uZ = (Z_1, \dots, Z_K) \in \R_+^K$ are the charges of fixed  nuclei located at $\uR = (R_1, \dots, R_K) \in \R^{3K}$.

For $N> 0$ (not necessarily integer valued)  and $Z_i \ge 0$, we now define the ground state energy in M\"uller theory by
\[
E_{\uR}(N, Z) = \inf \left\{\mathcal{E}_{\uR}(\gamma) \colon
\gamma \in \mathcal{P},
\tr \gamma = N\right\}
\]
where $\mathcal{P} = \{ \gamma \colon \gamma =\gamma^{\dagger}, 0 \le \gamma \le 1, (-\Delta +1)^{1/2}\gamma (-\Delta +1)^{1/2} \in \mathcal{S}^1\}$, $\mathcal{S}^1$ is the set of trace-class operators. 
When $N \le Z$, it was shown by Frank et. al.~\cite{frank2007muller} that $E_{\uR}(N, Z)$ has a minimizer.

In this paper, we will investigate minimization of the M\"uller energy over the nuclear positions $R_j$, that is, the Born-Oppenheimer energy of a molecule defined as
\begin{equation}
\label{inf}
E(N, \uZ) = \inf_{\uR} \left\{E_{\uR}(N, Z) + U_{\uR}\right\},
\end{equation}
where $U_{\uR}$ is the nuclear-nuclear repulsion
\[
U_{\uR} = \sum_{i < j}\frac{Z_iZ_j}{|R_i - R_j|}.
\]

Our purpose is to explore the existence of molecules in M{\"u}ller theories.
Following, we will say that the molecular system is stable if there exists a density-matrix $\gamma$ with $\tr \gamma = N$ such that $E(N, \uZ) = \mathcal{E}_{\uR}(\gamma) + U_{\uR} $ for some  $\uR \in \R^{3K}$.

Analogously to a series of works~\cite{CattoLions1, CattoLions2, CattoLions3, CattoLions4} by Catto and Lions on the Thomas-Fermi and Hartree type theories, we prove that any molecular system is stable under the M\"uller theory if and only if all possible two molecules can be bound.

It is well-known that, due to the classical work of Lieb and Thirring~\cite{LT1986}, neutral atoms and molecules are stable in the nonrelativistic Schr{\"o}dinger theory.
In particular, it was shown that the $R^{-6}$ attractive interaction energy, among molecules for large separation $R$, appears from the dipole-dipole interaction.
On the other hand, density-functional theory may not have the same feature, since it deals only with single particle densities, as pointed out in~\cite{LT1986}.
In Thomas-Fermi theory,  two neutral molecules can never be bound by Teller's no-binding theorem~\cite{LiebTF, LiebSimon1977}.
We refer to~\cite{CattoLions1, CattoLions2, CattoLions3, CattoLions4, BrisLions2005, LiebTF} for other Thomas-Fermi type theories and Hartree-Fock theories.
We recall M{\"u}ller theory is not a density functional but a density-matrix functional theory.
Namely, this theory describes the energy as a functional of the one-body density matrix $\gamma$, rather than a one-particle density $\rho$.
The first goal of this article is to extend the methods of~\cite{CattoLions1, CattoLions2, CattoLions3, CattoLions4} to investigate the M\"uller theory of molecules.

Let us define
\[
\Eh_{\uR}(\gamma) = \E_{\uR}(\gamma) + \frac{\tr \gamma}{8}.
\]

We note that
\[
-\frac{N}{8} = E_\infty(N) = \inf \left\{\E_\infty (\gamma) \colon \tr \gamma = N \right\}
\]
by~\cite[Propositon 1]{frank2007muller}, where
\[
 \E_\infty (\gamma) \coloneqq  \tr \left( -\frac{1}{2} \Delta \right) \gamma
+D[\rho_{\gamma}] -X(\gamma^{1/2}).
\]

For technical reason, we set a relaxed problem
\begin{equation}
\label{relax}
\eh_{\le}(N, \uZ) = \inf_{\uR} \left\{\eh_{\le}(N, Z, \uR)+U_{\uR}\right\},
\end{equation}
where
\[
\eh_{\le}(N, Z, \uR) = \inf \left\{\Eh_{\uR}(\gamma) \colon \gamma \in \mathcal{P}, \tr \gamma \le N\right\}.
\]
For any $N > 0$, $Z > 0$, it was shown in~\cite{frank2007muller}, $\eh_{\le}(N, Z, \uR)$ has a minimizer.

Our results are as follows.

\begin{thm}
\label{generaliff}
Any minimizing sequence $(\uR_n)_n \subset \R^{3K}$ for (\ref{relax}) is bounded if and only if
 \begin{equation}
\label{general}
\eh_{\le}(N, \uZ) < \eh(N_1, \underline{Z_1}) + \eh(N_2, \underline{Z_2})
\end{equation}
for all $N_i \ge 0$, $i=1, 2$, such that $N_1+N_2\le N$ and for any configuration $\underline{Z_1} = (Z_{j(1)}, \dots, Z_{j(p)})$ and $\underline{Z_2}= (Z_{j(p+1)}, \dots, Z_{j(K)})$, j permutation of $\{1, \dots, K \}$.
\end{thm}

As mentioned above, for $N \le Z$, a minimizer of M{\"u}ller energy has trace $N$.
Thus $\eh_{\le}(N, \uZ) = E(N, \uZ) + N/8$ and the molecules are stable when the binding inequality (\ref{general}) holds.
Moreover,

\begin{thm}
\label{reduce1}
We assume  $\eh_{\le}(N, \uZ) = E(N, \uZ) +N/8$. Then any minimizing sequence $(\uR_n)_n \subset \R^{3K}$ for (\ref{inf}) is bounded if and only if
 \begin{equation}
\label{genreduce}
E(N, \uZ) < E(N_1, \underline{Z_1}) + E(N_2, \underline{Z_2})
\end{equation}
for all $N_i \ge 0$, $i=1, 2$, such that $N_1+N_2 = N$ and for any configuration $\underline{Z_1} = (Z_{j(1)}, \dots, Z_{j(p)})$ and $\underline{Z_2}= (Z_{j(p+1)}, \dots, Z_{j(K)})$, j permutation of $\{1, \dots, K \}$.
\end{thm}

It is expected that  binding occurs for $N\le Z$ molecules or ions, though it is an open question.
Even in the Hartree-Fock theory, the existence of molecules is still open except in special cases~\cite{CattoLions1, CattoLions2, CattoLions3, CattoLions4, BrisLions2005}.

One main purpose of this article is as follows.

\begin{thm}[Bound on the  positive excess charge]
\label{th:ionization}
We assume $N \le c_1Z$ and $Z_\mathrm{min} \coloneqq \min\{Z_1, \dots, Z_K \} \ge c_2Z$ with some constants $c_i > 0$, $i=1,2$, independent of $Z$.
If there exist a stable configuration $\uR = (R_1, \dots, R_K) \in \R^{3K}$ and a density matrix $\gamma \in \mathcal{P}$  such that $\E_{\uR}(\gamma)+U_{\uR}  = E(N, \uZ) $,  then there exist $C_0>0$ depending only on $Z_1, \dots, Z_K$, and $K$, $c_i>0$ such that
\begin{equation}
\label{eq:ionization}
Z - N \le C_0Z^{1-\delta}
\end{equation}
for some $\delta>0$.

Moreover, if we put $R_\mathrm{min} \coloneqq \min_{i \neq j} |R_i - R_j|$, then there is a constant $C>0$ depending on the same quantities as above $C_0$ so that
\begin{equation}
\label{th:radius}
R_\mathrm{min} > CZ^{-(1/3)(1-\epsilon)},
\end{equation}
where $\epsilon = 2/77$. 
\end{thm}

\begin{rem}
It is expected that if a M\"uller minimizer exists, then $N \le CZ$ holds.
In fact, for the atomic case, if there is a minimizer then $N \le Z + \mathrm{const}$ holds~\cite{mullerIC}.
However, the proof  works only for the atomic case, and it is still an open issue for molecular cases.
\end{rem}

\begin{rem}
The  estimate (\ref{th:radius}) states that the molecular radii in the frame work of M\"uller theory are much larger than the Thomas-Fermi atomic radii, namely $Z^{-1/3}$.
Thus the Thomas-Fermi density of the molecule is of order of the sum of atomic densities.
Solovej and Ruskai~\cite{RuskaiSolovej, SolovejAN} showed by using this type estimate that the asymptotic neutrality $N-Z = o(Z)$ for molecules in nonrelativistic  Schr\"odinger theory.
\end{rem}

\section{di-atomic case}
First, we consider a simple di-atomic case. Without loss of generality, we may assume
\[
V_{\uR}(x) = V_R(x) = \frac{Z_1}{|x|} + \frac{Z_2}{|x-R \hat e|}, \quad U_{\uR} = U_R = \frac{Z_1Z_2}{R},
\]
where $R > 0$, and $\hat e \in \R^3$ is an unit vector.
Then our minimization problem is

\begin{equation}
\label{diatomic}
\eh_{\le}(N, Z) = \inf_{R>0} \left\{\eh_{\le}(N, Z, R)+\frac{Z_1Z_2}{R} \right\}.
\end{equation}

In this section our goal is

\begin{thm}
\label{iff}
Any minimizing sequence for (\ref{diatomic}) is bounded if and only if

\begin{equation}
\label{binding}
\eh_{\le}(N, Z) < \eh_\mathrm{atom}(N_1, Z_1) + \eh_\mathrm{atom} (N_2, Z_2),
\end{equation}
for all $0\le N_i$, $i=1, 2$, such that $N_1+N_2\le N$.
Here 
\[
\eh_\mathrm{atom}(N, Z) = \inf \{ \Eh_\mathrm{atom}(\gamma) \colon \gamma \in \mathcal{P}, \tr \gamma = N\},
\]
 and
\[
\Eh_\mathrm{atom}(\gamma ) = \tr \left( -\frac{1}{2} \Delta -Z|x|^{-1} \right) \gamma
+D[\rho_{\gamma}] -X(\gamma^{1/2}) +\frac{\tr \gamma}{8}.
\]
\end{thm}

The next Lemma corresponds to the `only if' part of the theorem.

\begin{lem}
\label{if}
For all $N_i \ge 0$, $i=1, 2$, with $N_1+N_2\le N$, we have
\begin{equation}
\begin{split}
\label{diatomicineq}
\eh_{\le}(N, Z) &\le \limsup_{R \to \infty} (\eh_{\le}(N, Z, R) + U_R) \\
&\le \eh_\mathrm{atom}(N_1, Z_1) + \eh_\mathrm{atom} (N_2, Z_2).
\end{split}
\end{equation}
\end{lem}

It immediately follows that

\begin{cor}
\label{ifcor}
We assume $\eh(N, Z) = \eh_\le(N, Z)$. For all $N_i \ge 0$, $i=1, 2$, with $N_1+N_2\le N$, we have
\begin{equation*}
\begin{split}
\label{diatomicineq2}
E(N, Z) &\le \limsup_{R \to \infty} (E_\le(N, Z, R) + U_R) \\
&\le E_\mathrm{atom}(N_1, Z_1) + E_\mathrm{atom} (N_2, Z_2).
\end{split}
\end{equation*}
\end{cor}

We shall prove Lemma~\ref{if}.
The following lemma is obtained by the same proof in~\cite[Lemma 1]{Kehle2017}.

\begin{lem}
Let $Z \ge 0$, $N > 0$ and $\tr \gamma = N$. Then, for any $\epsilon > 0$ there exists a density-matrix $\sigma \in \mathcal{P}$ having a compactly supported integral kernel, $\tr \sigma = N$ and
\[
|\mathcal{E}_{R} (\gamma) - \mathcal{E}_{R} (\sigma)| \le \epsilon.
\]
\end{lem}

\begin{proof}[Proof of Lemma \ref{if}]
It is trivial for $N_1 = 0$ (or equivalently, $N_2 = 0$).
Let $\epsilon > 0$, $N_i > 0$, $i=1, 2$, and $N_1 + N_2 \le N$.
We may assume
$\Eh_\mathrm{atom}(\gamma_i) \le \eh_\mathrm{atom}(N_i, Z_i) + \epsilon/3$,
$\tr\gamma_i =  N_i$, and each kernel of $\gamma_i$ is compactly supported in a ball with the radius $r > 0$.
Let $\widehat{\gamma_2}_{R} = \tau_{-R} \gamma_2 \tau_{R}$ with $\tau$ being the translation operator.
We then define a trial density-matrix by
\[
\gamma_{R} = \gamma_1 + \widehat \gamma_{2_R}.
\]
Clearly $0 \le \gamma \le 1$, $\tr \gamma \le N$, and $\gamma_1  \widehat \gamma_{2_R} = 0$ for large $R$, by construction.
Thus we can compute $X(\gamma_R^{1/2}) = X(\gamma_1^{1/2})$ + $X(\widehat \gamma_{2_R}^{1/2})$.
Furthermore, it is easy to see that
\begin{equation*}
2D[\rho_{\gamma_1}, \rho_{\widehat{\gamma_2}_R}] = \iint_{\R^3 \times\R^3} \,\frac{\rho_{\gamma_1}(x) \rho_{\widehat{\gamma_2}_R}(y)}{|x-y|} \, dx \, dy
\le \frac{N_1N_2}{R - 2r}.
\end{equation*}
Using the translation invariant of the functional $\E_\infty (\gamma)$, we may find
\begin{equation*}
\begin{split}
 \eh_{\le}(N, Z, R) +\frac{Z_1Z_2}{R}  
 &\le
\Eh_R(\gamma_R) +\frac{Z_1Z_2}{R} \\
&\le \sum_{i=1, 2}\Eh_\mathrm{atom}(\gamma_i) + 2D(\rho_{\gamma_1}, \rho_{\widehat{\gamma_2}_R}) + \frac{Z_1Z_2}{R} \\
&\le \sum_{i=1, 2}\eh_\mathrm{atom}(N_i, Z_i) + \frac{2\epsilon}{3} + \frac{N_1N_2}{R - 2r} + \frac{Z_1Z_2}{R},
\end{split}
\end{equation*}
for sufficiently large $R> 0$.
Hence for any given $\epsilon > 0$ and $N_1 + N_2 \le N$, it holds that

\begin{align*}
\limsup _{R\to \infty} &\left(\eh_\le(N, Z, R) +\frac{Z_1Z_2}{R}\right) \\
 &\le \eh_\mathrm{atom}(N_1, Z_1) + \eh_\mathrm{atom} (N_2, Z_2) +  \epsilon,
\end{align*}
which shows (\ref{diatomicineq}).
\end{proof}

Lemma \ref{if} implies that if any minimizing sequence $(R_n)_n$ for (\ref{diatomic}) is bounded, then the binding inequality (\ref{binding}) holds.
Indeed, suppose
$\eh(N, Z) = \eh_\mathrm{atom}(N_1, Z_1) + \eh_\mathrm{atom} (N_2, Z_2)$ for some $N_1 + N_2 \le N$. Then, by Lemma \ref{if}, $\lim_{R \to \infty} (\eh_\le(N, Z, R) + U_R) = \eh(N, Z)$. 
This contradicts to the assumption that any minimizing sequence is bounded.
Hence, the `only if' part of Theorem \ref{iff} is followed.

\begin{proof}[Proof of Theorem \ref{iff}]
It suffices to show  the `if' part.
We suppose that there is a minimizing sequence $(R_n)_n$ for $\eh_{\le}(N, Z)$ so that $R_n \to \infty$.
Then we may assume that there exist density-matrices $\gamma_n \in \mathcal{P}$ so that $\Eh_{R_n}(\gamma) + U_{R_n} \to \eh_{\le}(N, Z)$ as $n \to \infty$.
Using the hydrogen bound, it follows that
\[
\tr Z_j|x-R_j|^{-1} \gamma \le \frac{Z_j \epsilon}{4Z} \tr (-\Delta) \gamma + \frac{Z_jZ}{\epsilon} \tr \gamma,
\]
for any positive number $\epsilon >0$.
Hence $\tr V_R \gamma \le \epsilon/4 \tr(-\Delta \gamma) + Z^2/\varepsilon \tr \gamma$, for any $\epsilon > 0$.
Moreover,  the hydrogen bound also implies that

\begin{lem}[Lemma 1 of~\cite{frank2007muller}]
\label{exchangeineq}
For any $\epsilon >0$ it holds that
\begin{equation*}
X(\gamma^{1/2}) \le \frac{\epsilon}{4} \tr (- \Delta \gamma) + \frac{1}{4 \epsilon} \tr \gamma.
\end{equation*}
\end{lem}
Now we get the following bound as~\cite[Equation (57)]{frank2007muller}:
\begin{equation*}
\frac{1}{2} (1- \epsilon)\tr (- \Delta)\gamma_n
\le \Eh_{R_n}(\gamma_n) + U_{R_n} + \frac{1}{\epsilon} \left(Z^2 + \frac{1}{4} \right) \tr \gamma_n.
\end{equation*}
Hence $(-\Delta +1)^{1/2} \gamma_n (-\Delta +1)^{1/2}$ is bounded in $\mathcal{S}^1$, and thus, by the Banach-Alaoglu theorem, after passing to a subsequence if necessary we may assume that $\tr K \gamma_n \to \tr K \gamma$ for some $\gamma$ and for any operator $K$ such that $(-\Delta +1)^{1/2} K (-\Delta +1)^{1/2}$ is compact.
In particular, for any function $f \in L^p(\R^3)$ $(3/2 \le p < \infty)$
\begin{equation*}
\int_{\R^3} f(x) \rho_{\gamma_n}(x) \,dx
 = \tr f \gamma_n
\to \tr f \gamma
= \int_{\R^3} f (x) \rho_\gamma(x) \,dx.
\end{equation*}
We note that $0 \le \gamma \le 1$ and
\begin{equation}
\label{semicont}
M = \tr \gamma \le \liminf_{n \to \infty} \tr \gamma_n = \widetilde N \le N
\end{equation}
by the lower-semicontinuity of the $\mathcal{S}^1$ norm.

We see  $\gamma \not \equiv  0$ from~\cite[Proposition 1]{frank2007muller}.
In fact, for some $\delta > 0$
\begin{equation*}
\eh_\mathrm{atom}(N, Z_1) \le -\delta.
\end{equation*}
From Lemma \ref{if},
\[
\limsup _{R\to \infty} \eh_{R}(N, Z) \le \eh_\mathrm{atom}(N, Z_1).
\]
Thus, 
$\Eh_{R_n}(\gamma_n) + U_{R_n} \le -\epsilon$ for some $\epsilon > 0$ and sufficiently large $n$.
Hence, we have
\[
-\epsilon \ge \Eh_{R_n}(\gamma_n) + U_{R_n}  \ge - \tr V_{R_n} \gamma_n,
\]
and thus
\[
\tr V_{R_n} \gamma_n \ge \epsilon,
\]
where $V_{R_n} = Z_1|x|^{-1} + Z_2|x-R_n \hat e|^{-1}$.
Thus $\gamma \not \equiv 0$.

If $M = \widetilde N$, then $\lim_{n \to \infty}\tr \gamma_n = \tr \gamma$. Thus $\gamma_n \to \gamma$ as $n \to \infty$ in $\mathcal{S}^1$ by ~\cite[Theorem A.6]{Simon1979T}.
Then 
\begin{equation}
\label{infty}
\int_{\R^3} \rho_{\gamma_n} (x)|x-R_n \hat e|^{-1} dx \to 0
\end{equation}
by $R_n \to \infty$.
Indeed, we may split
\begin{equation*}
\begin{split}
 \int_{\R^3} \frac{\rho_{\gamma_n}(x)}{|x-R_n \hat e|} \,dx
 &= \int_{\R^3}\left( \frac{\rho_{\gamma_n}(x) - \rho_{\gamma}(x)}{|x-R_n \hat e|} +\frac{\rho_{\gamma}(x)}{|x-R_n \hat e|}\right)  \,dx.
\end{split}
\end{equation*}
We see that the second term converges to $0$ by Young's inequality.
For the first term, we split $\rho_{\gamma_n}(x) - \rho_{\gamma}(x) = (\sqrt{\rho_{\gamma_n}(x)} + \sqrt{\rho_{\gamma}(x)} )(\sqrt{\rho_{\gamma_n}(x)} - \sqrt{\rho_{\gamma}(x)} )$.
We know that $\sqrt{\rho_{\gamma_n}} \to \sqrt{\rho_{\gamma}}$ strongly in $L^2(\R^3)$ by $\gamma_n \to \gamma$ in $\mathcal{S}^1$, and thus the first term also converges to $0$.

From the lower-semicontinuity of our functionals (see~\cite[Proposition 3]{frank2007muller}),
we have
\begin{equation*}
\label{contra}
\eh_{\le}(N, Z) \ge 
\liminf_{n \to \infty} \Eh_\mathrm{atom}(\gamma_n) \ge
\Eh_\mathrm{atom}(\gamma) \ge
\eh_\mathrm{atom}(\widetilde N, Z_1) \ge \eh_{\le}(N, Z),
\end{equation*}
and thus $\eh_{\le}(N, Z) = \eh_\mathrm{atom}(\widetilde N, Z_1)$ with $\widetilde N \le N$.
Then we have finished the proof in this case.

Let 
\[
(\chi^0)^2 + (\chi^1)^2 = 1
\]
with $\chi^0 \in C^{\infty}(\R^3)$, radial, $\chi^0(0)=1$, $\chi^0(r) < 1$ if $r > 0$, $\chi^0(r) = 0$ if $r \ge2$.
For each $j$ $\tr (\chi^0(|x|/L) )^2 \gamma_j$ is a continuous function of $L> 0$ which increases from $0$ to $\tr \gamma_j$.
Now $\tr \gamma_j > M$ for large $j$, and thus we can choose $L_j$ such that $\tr \gamma_j^0 \coloneqq \tr (\chi^0 (|x|/L_j))^2\gamma = M$, $L_j \to \infty$, and then $\gamma_j^0 \to \gamma$ in $\mathcal{S}^1$.
We write $\chi^\nu_j (x/L_j) \coloneqq \chi^\nu (|x|/L_j)$ and $\gamma_j^\nu = \chi^\nu_j \gamma_j \chi_j^\nu$ for each $\nu = 0, 1$.

From the IMS formula,
\[
\tr (-\Delta \gamma_n) = \sum_{\nu = 0, 1} \left[\tr (-\Delta \gamma_N^\nu) - 
\tr |\nabla\chi_n^\nu|^2 \gamma_n\right].
\]
Clearly,
\[
D[\rho_{\gamma_j}] = D[\rho_{\gamma_j^0}] + D[\rho_{\gamma_j^1}]+  2D(\rho_{\gamma_j^0}, \rho_{\gamma_j^1})
\ge D[\rho_{\gamma_j^0}] + D[\rho_{\gamma_j^1}]
\]
since $\rho^\nu_{\gamma_j} \ge 0$.
For the potential term, we learn
\[
\tr (|x|^{-1}\gamma_n) = \tr (|x|^{-1}\gamma_n^0) + o(1)
\]
and
\[
\tr (|x-R_n \hat e|^{-1}\gamma_n) = \tr (|x-R_n \hat e|^{-1}\gamma_n^1) + o(1),
\]
because $R_n \to \infty$ as (\ref{infty}).
For the exchange term,  we have $X(\gamma_j^{1/2}) \le X(({\gamma_j^0})^{1/2}) + X(({\gamma_j^1})^{1/2}) + o (1)$ as~\cite{frank2007muller}.
Let $\widetilde \gamma_n = \tau_{-R_n \hat e} \gamma_n^1 \tau_{R_n \hat e}$.
It is clear that $\tr \widetilde \gamma_n = K - M$ with some $K \le N$.
By the translation invariant for the functional $\E_\infty(\gamma)$, we have
\begin{equation*}
\begin{split}
\label{conclusion}
\Eh_{R_n}(\gamma_n) + U_{R_n} &\ge \Eh_\mathrm{atom}(\gamma_n^0) + \Eh_\mathrm{atom}(\widetilde \gamma_n) + o(1) \\
&\ge \Eh_\mathrm{atom}(\gamma_n^0) + \eh_\mathrm{atom}(K-M; Z_2) + o(1) .
\end{split}
\end{equation*}
Hence, again by the lower-semicontinuity, we arrive at
\begin{equation*}
\begin{split}
\eh_{\le}(N, Z) &\ge \liminf_{n \to \infty}\Eh_\mathrm{atom}(\gamma_n^0) + \eh_\mathrm{atom}(K-M, Z_2) \\
&\ge \Eh_\mathrm{atom}(\gamma) + \eh_\mathrm{atom}(K-M, Z_2).
\end{split}
\end{equation*}
Thus $\eh_{\le}(N, Z)  \ge \eh_\mathrm{atom}(M, Z_1) + \eh_\mathrm{atom}(K-M, Z_2)$ with $K \le N$.
Using Lemma \ref{if}, we have the theorem.
\end{proof}

We recall $E_\infty(N) = -N/8$.
The next theorem which corresponds to Theorem~\ref{reduce1} follows.

\begin{thm}
\label{dired1}
We assume $\eh_{\le}(N, Z) = E(N, Z) +N/8$. 
Then, any minimizing sequence for (\ref{inf}) is bounded if and only if

\begin{equation}
\label{dired2}
E(N, Z) < E_\mathrm{atom}(N_1, Z_1) + E_\mathrm{atom} (N_2, Z_2)
\end{equation}
for all $N_1 + N_2 = N$, $0 \le N_i$, $i=1, 2$.
\end{thm}

\begin{proof}[Proof of Theorem~\ref{dired1}]
In the proof of the previous theorem, we may take $K = N$ when $\eh_{\le}(N, Z) = E(N, Z) +N/8$.
Thus molecules are stable if and only if (\ref{binding}) holds for all $N_1 + N_2 = N$.
Then, the binding inequalities (\ref{dired2}) and (\ref{binding}) are equivalent for $N_1 + N_2 = N$.
\end{proof}

\section{General case}
First, we show the following proposition.

\begin{prop}
\label{ifpart}
It is always the case
 \begin{equation*}
\eh_{\le}(N, \uZ) \le \eh_{\le}(N_1, \underline{Z_1}) + \eh_{\le}(N_2, \underline{Z_2})
\end{equation*}
for all $N_i \ge 0$, $i=1, 2$, such that $N_1+N_2\le N$.
\end{prop}

\begin{proof}[Proof of Proposition \ref{ifpart}]
Let $\epsilon > 0$.
As the proof of Lemma \ref{if}, we can take $\gamma_i^n$ and $\underline{R_i^n}$, $i=1, 2$, such that
\[
\Eh_{\underline{R_i^n}} (\gamma_i^n) + U_{\underline{R_i^n}} \le \eh_{\le}(N_i, \underline{Z_i}) + \frac{1}{n}, 
\]
$i=1, 2$. Moreover, we may assume that each kernel of $\gamma_i^n$ has a compact support in a ball.
Let $\hat \gamma_2^n = \tau_{-B_n} \gamma_2^n \tau_{B_n}$, with $B_n \in \R^3$.
We define $\gamma^n = \gamma_1^n + \hat \gamma_2^n$ as the diatomic case.
Then, for $\uR^n = (R_{j(1)}^n, \dots, R_{j(p)}^n, R_{j(p+1)}^n +B_n, R_{j(p+2)}^n +B_n, \dots, R_{j(K)}^n + B_n)$ with large $|B_n|$,
\begin{align*}
\eh_{\le}(N, \uZ) &\le \eh_{\uR^n}(N, Z) + U_{\uR^n}
\le \Eh_{\uR^n}(\gamma^n) +U_{\uR^n}\\
&\le \eh_{\le}(N_1, \underline{Z_1}) + \eh_{\le}(N_2, \underline{Z_2}) + \epsilon
\end{align*}
for sufficiently large $n$.
\end{proof}

\begin{rem}
It immediately follows that the `only if' part of Theorem \ref{general} by this proposition.
Suppose that any minimizing sequence for (\ref{relax}) is bounded in $\R^{3K}$.
If  $\eh_{\le}(N, \uZ) = \eh_{\le}(N_1, \underline{Z_1}) + \eh_{\le}(N_2, \underline{Z_2})$ for some configuration, then the above $\uR^n$ is a minimizing sequence and clearly not bounded.
\end{rem}

\begin{proof}[Proof of Theorem \ref{generaliff}]
We only show the `if' part by contradiction.
Let $\Eh_{\uR^n}(\gamma_n) + U_{\uR^n} \to \eh_{\le}(N, \uZ)$ and suppose this $\uR^n$ is not bounded.
As the proof of the di-atomic case, we may assume $\gamma_n \to \gamma \not \equiv 0$ in a sense, and the relation (\ref{semicont}) holds.
If $\tr \gamma = M = \widetilde N$, then $\gamma_n \to \gamma$ in $S^1$.
Then, after passing by subsequence if necessary,
\begin{equation*}
\eh_{\le}(N, \uZ) \ge \liminf_{n\to \infty} (\Eh_{\uR^n}(\gamma_n) + U_{\uR^n}) \ge \Eh_{\uR}(\gamma),
\end{equation*}
where $\uR \in \R^{3(K - L)}$, $L$ is the number of $i$ such that $|R_i^n| \to \infty$.
Hence $\eh_{\le}(N, \uZ) = \eh(\widetilde N, \widetilde \uZ)$ with $\widetilde N \le N$ and thus $\eh_{\le}(N, \uZ) \ge \eh_{\le}(N, \widetilde \uZ)$.

Next, we consider the case of $M < \widetilde N$.
We may split $\gamma_n = \gamma_n^0 + \gamma_n^1$, $\gamma_n^0 \to \gamma$ in $S^1$.
Let $J = \{j \colon R_j^n \text{ remain bounded} \}$.
If $J = \emptyset$, passing to a subsequence if necessary, we may $|R_j^n| \to \infty$ for all $j$.
Then, 

\begin{equation*}
\tr (|x-R_j^n|^{-1} \gamma_n) = \tr (|x-R_j^n|^{-1} \gamma_n^1) + o(1).
\end{equation*}
as the same reason of (\ref{infty}).
Thus we get
\begin{equation*}
\begin{split}
\Eh_{\uR^n}(\gamma_n) + U_{\uR^n} &\ge \Eh_\infty(\gamma_n^0) + \Eh_{\uR^n}(\gamma_n^1) + U_{\uR_n}+ o(1) \\
&\ge \Eh_\infty(\gamma_n^0) + \eh(K-M, \uZ) + o(1) .
\end{split}
\end{equation*}
Thus
\begin{equation*}
\eh_{\le}(N, \uZ) \ge \eh_\infty(M) + \eh(K - M, \uZ).
\end{equation*}

If $J \neq \emptyset$, then, by passing to a subsequence if necessary, we may assume that $R_j^n  \to R_j$ for $j \in J$ and $|R_i| \to \infty$ for $i \not \in J$.
Then, for $j \in J$ we see that
\begin{equation*}
\tr (|x-R_j^n|^{-1} \gamma_n) = \tr (|x-R_j|^{-1} \gamma_n^0) + o(1).
\end{equation*}
For $j \not \in J$,
\begin{equation*}
\tr (|x-R_j^n|^{-1} \gamma_n) = \tr (|x-R_j^n|^{-1} \gamma_n^1) + o(1).
\end{equation*}
Hence we arrive at
\begin{equation*}
\eh_{\le}(N, Z) \ge \eh(M, \underline{Z_1}) + \eh(K - M, \underline{Z_2}),
\end{equation*}
where $\underline{Z_1} = (Z_j)_{ j \in J}$ and $\underline{Z_2} = (Z_j )_{j \not \in J}$. 
This completes the proof.
\end{proof}

We now turn to the

\begin{proof}[Proof of Theorem \ref{reduce1}]
If $\eh_\le(N, Z) = E_\le(N, Z) + N/8$, then we can take $K=N$ in the above proofs.
Therefore, any minimizing sequence is bounded if and only if the binding condition (\ref{general}) holds for all $N_1 + N_2 = N$.
For $N_1 + N_2 = N$ the conditions (\ref{general}) and (\ref{genreduce}) are equivalent.
Thus Theorem \ref{reduce1} follows.
\end{proof}

\section{A lower bound on the size of molecules}
\label{size}
In this section we prove the estimate (\ref{th:radius}) in Theorem \ref{th:ionization}.
First, we use the united atom bound for M\"uller theory.

\begin{prop}[United atom bound]
For any $N >0$ and for any configuration $\uR \in \R^{3K}$ we have
\[
E_{\uR}(N, Z) \ge E_\mathrm{atom}(N, Z).
\]
\end{prop}

\begin{proof}
Let $\epsilon > 0$ and $E_{\uR}(N, Z) \ge \E_{\uR}(\gamma) + \epsilon$.
Then we note
\[
 \E_{\uR}(\gamma)
=\sum_{j=1}^K \frac{Z_j}{Z}\left[ \tr \left( -\frac{1}{2} \Delta - Z|x-R_j|^{-1} \right) \gamma
+D[\rho_{\gamma}] -X(\gamma^{1/2})\right].
\]
Since the energy of $E_\mathrm{atom}(N, Z)$ is independent of nucler positions $R_j$, the conclusion follows. 
\end{proof}
From this bound we have
\[
E(N, \uZ) \ge E_\mathrm{atom}(N, Z) + \frac{Z_iZ_j}{|R_i-R_j|}
\]
where $R_\mathrm{min} = |R_i-R_j|$.
We now deduce from Lemma~\ref{exchangeineq} that
\[
E_\mathrm{atom}(N, Z) \ge \tr \left( -\frac{1}{4}\Delta \right) \gamma -\tr( Z|x|^{-1})\gamma + D[\rho_\gamma] - \frac{N}{4}.
\]
We need the
\begin{thm}[Lieb-Thirring kinetic energy inequality \cite{LS}]
\begin{equation*}
\tr \left(-\frac{\Delta}{2} \gamma \right) \ge \frac{3}{10}L\int_{\R^3} \rho_{\gamma}(x)^{5/3} \, dx,
\end{equation*}
with a constant $L$ (see~\cite{Dolbeault, FHJN}).
\end{thm}
Hence we infer that
\[
E_\mathrm{atom}(N, Z) 
\ge  \frac{3}{10}C\int_{\R^3} \rho(x)^{5/3} \, dx - \tr (Z|x|^{-1}) \gamma+D[\rho] - \frac{N}{4}.
\]
Next, we introduce the Thomas-Fermi (TF) functional \cite{LiebTF, LiebSimon1977} by 
\begin{equation*}
\E_\mathrm{\uR}^{TF}(\rho) = \frac{3}{10}(3 \pi^2)^{2/3}A \int_{\R^3} \rho(x)^{5/3} \, dx + \int_{\R^3} V_R(x) \rho(x) \, dx +D[\rho],
\end{equation*}
and define the lowest energy by
\begin{equation*}
E_{\uR}^{TF}(N, Z, A) = \inf \left\{ \E_{\uR}^{TF}(\rho) \colon 0 \le \rho, \int_{\R^3} \rho (x) \, dx = N, \rho \in L^{5/3}(\R^3) \right\}.
\end{equation*}
From the scaling property of the Thomas-Fermi functional~\cite{LiebTF}, we see $E_\mathrm{atom}^\mathrm{TF}(N, Z, A) \ge  -CZ^{7/3}$.
Consequently, we arrive at
\[
E(N, Z) \ge -CZ^{7/3} + \frac{Z_iZ_j}{|R_i-R_j|}.
\]
Hence we have $|R_i-R_j| \ge CZ^{-1/3}$.

Next, we  improve this bound by comparison with Thomas-Fermi theory.
In order to compare our functional with Thomas-Fermi one, we  need the following semiclassical approximations.
The following results are taken from~\cite[Lemma 8.2]{SolovejIC} (we use the optimal $\delta>0$ as in~\cite[Lemma 11]{mullerIC}).

\begin{lem}
\label{semicl}
For fixed $s>0$ and smooth $g \colon \R^3 \to [0, 1]$ satisfying $\mathrm{supp} \, g \subset \{|x|<s\}$, $\int g^2=1$, $\int |\nabla g|^2 \le Cs^{-2}$ it follows that

\begin{enumerate}
\item[(i)] For any $V \colon \R^3 \to \R$ with $[V]_+$, $[V - V \star g^2]_+ \in L^{5/2}$ and for any $0 \le \gamma \le 1$
\begin{equation*}
\begin{split}
\tr \left(- \frac{\Delta}{2} -V\right)\gamma &\ge -2^{5/2}(15 \pi^2)^{-1}\int[V]^{5/2}_+ - Cs^{-2} \tr \gamma \\
& \quad -C\left( \int [V]^{5/2}_+ \right)^{3/5} \left(\int [V-V \star g^2]^{5/2}_+\right)^{2/5},
\end{split}
\end{equation*}
where the symbol $[x]_+$ stands for $\max\{0, x\}$.

\item[(ii)]
If $[V]_+ \in L^{5/2} \cap L^{3/2}$, then there is a density-matrix $\gamma$ so that $\rho_\gamma = 2^{5/2}(6\pi^2)^{-1}[V]^{3/2}_+ \star g^2$,
\[
\tr \left( -\frac{\Delta}{2}\gamma \right) \le 2^{3/2}(5\pi^2)^{-1} \int [V]_+^{5/2}
+ Cs^{-2}\int[V]_+^{3/2}.
\]
\end{enumerate}
\end{lem}
We introduce the TF potential for the molecule as the function
\[
\phi_\mathrm{mol}^\mathrm{TF}(x) \coloneqq \sum_{i=1}^KZ_i|x-R_i|^{-1} - 
\int_{\R^3} \frac{\rho_\mathrm{mol}^\mathrm{TF} (y)}{|x-y|} \, dy,
\]
where $\rho_\mathrm{mol}^\mathrm{TF}$ is the unique minimizing density for $E^\mathrm{TF}(N, Z, \uR) = E^\mathrm{TF}(N, Z, \uR, 1)$ (when $N > Z$ we take the minimizer for the neutral molecule).
First, we shall show 

\begin{lem}
\label{lem:lower}
For any configuration $\uR \in \R^{3K}$ and density-matrix $\gamma$ we have

\begin{equation}
\label{lower}
\E_{\uR}(\gamma) \ge \E_{\uR}^\mathrm{TF}(\rho_\mathrm{mol}^\mathrm{TF}) + D\left[\rho_\gamma - \rho_\mathrm{mol}^\mathrm{TF}\right] -CZ^{25/11}.
\end{equation}
\end{lem}

\begin{proof}[Proof of Lemma \ref{lem:lower}]
We can write
\[
\E_{\uR}(\gamma)  = \tr \left(-\frac{\Delta}{2} - \phi_\mathrm{mol}^\mathrm{TF}\right)\gamma + D\left[\rho_\gamma - \rho_\mathrm{mol}^\mathrm{TF}\right]  - D\left[\rho_\mathrm{mol}^\mathrm{TF}\right] -X(\gamma^{1/2}).
\]
According to $N \le CZ$, we may bound the exchange term by
\[
X(\gamma^{1/2}) \le CZ^{5/3}.
\]
Indeed, we infer from Hardy's inequality that
\begin{align*}
&\iint_{\R^3 \times \R^3} \frac{|\gamma^{1/2}(x, y)|^2}{|x-y|} \, dx \, dy  \\
& \le \left(\iint |\gamma^{1/2}(x, y)|^2\, dx \, dy\right)^{1/2}\left(\iint \frac{|\gamma^{1/2}(x, y)|^2}{|x-y|^2} \, dx \, dy\right)^{1/2} \\
& \le 4N^{1/2} (\tr (-\Delta) \gamma)^{1/2}.
\end{align*}
We recall $\tr (-\Delta) \gamma \le CZ^{7/3}$ by the energy bound.

Next, from Lemma \ref{semicl} (i) we have
\begin{equation*}
\begin{split}
&\tr \left(- \frac{\Delta}{2} -\phi_\mathrm{mol}^\mathrm{TF} + \mu(N, Z, \uR)\right)\gamma \\
&\ge -2^{5/2}(15 \pi)^{-1}\int[\phi_\mathrm{mol}^\mathrm{TF} -\mu(N, Z, \uR)]^{5/2}_+ - Cs^{-2}\tr \gamma \\
&-C\left( \int [\phi_\mathrm{mol}^\mathrm{TF}- \mu(N, Z, \uR)]^{5/2}_+ \right)^{3/5} \left(\int [\phi_\mathrm{mol}^\mathrm{TF}-\phi_\mathrm{mol}^\mathrm{TF} \star g^2]^{5/2}_+\right)^{2/5}.
\end{split}
\end{equation*}
Here $\mu(N, Z, \uR) \ge 0$ is the chemical potential for the molecule.
It is known~\cite{LiebTF} that the functions $\rho_\mathrm{mol}^\mathrm{TF}$ and $\phi_\mathrm{mol}^\mathrm{TF}$ satisfy the TF equation
\begin{equation}
\label{TFeq}
\rho_\mathrm{mol}^\mathrm{TF}(x)^{2/3} = 2^{5/3}(6\pi^2)^{-2/3}[\phi_\mathrm{mol}^\mathrm{TF}(x) - \mu(N, Z, \uR)]_+.
\end{equation}
Using the TF equation and scaling property in Thomas-Fermi theory, we have
\[
\int [\phi_\mathrm{mol}^\mathrm{TF}- \mu(N, Z, \uR)]_+^{5/2} \le C \int (\rho_\mathrm{mol}^\mathrm{TF})^{5/3} \le C{Z^{7/3}}.
\]
Since $V_R$ is superharmonic, it follows that  $V_{\uR} - V_{\uR} \star g^2 \ge 0$ by the maximum principle.
To see this, we note that $V_{\uR}  \star g^2$ is a continuous function going to zero at infinity, and therefore $\psi \coloneqq V_{\uR}  - V_{\uR} \star g^2 \to \infty$ as $x \to R_i$ for any $i$.
Since $\psi$ is continuous away from the $R_i$, $A \coloneqq \{ x \colon \psi(x) < 0\}$ is open and disjoint from the $R_i$.
Thus  $-\Delta \psi \le 0$ on $A$. 
It is clear that $\psi(x) \to 0$ as $|x| \to \infty$ and hence $A$ is empty by the maximum principle.
 Hence $\psi \ge 0$. 
 
Similarly, we see $|x|^{-1} - |x|^{-1}\star g^2 \ge0$, and hence $\rho_\mathrm{mol}^\mathrm{TF}\star |x|^{-1} - \rho_\mathrm{mol}^\mathrm{TF}\star g^2 \star |x|^{-1} \ge 0$ follows.
 We recall  Newton's theorem
 \[
 \int_{\S^2} |x-y|^{-1} \frac{d \nu (y)}{4\pi} = \min (|x|^{-1}, |y|^{-1})
 \]
 for any $x \in \R^3$.
 Then
\begin{equation*}
V_{\uR} - V_{\uR}\star g^2 \le \sum_{j=1}^KZ_j\left(|x-R_j|^{-1} \1(|x-R_j| \le s)\right).
\end{equation*}

Using this bound,  we obtain
\begin{equation*}
\begin{split}
\int [\phi_\mathrm{mol}^\mathrm{TF}-\phi_\mathrm{mol}^\mathrm{TF} \star g^2]_+^{5/2} 
&\le \int [V_{\uR} - V_{\uR}  \star g^2]_+^{5/2} \\
& \le Z^{5/2}\sum_{i=1}^K\int_{|x-R_i|\le s} |x-R_i|^{-5/2}\, dx \\
&\le CZ^{5/2}s^{1/2},
\end{split}
\end{equation*}
where we have used the convexity of $x^{5/2}$.
Hence
\begin{equation*}
\begin{split}
\tr \left(- \frac{\Delta}{2} -\phi_\mathrm{mol}^\mathrm{TF}\right)\gamma 
&\ge -2^{5/2}(15 \pi^2)^{-1}\int[\phi_\mathrm{mol}^\mathrm{TF}-\mu(N, Z, \uR)]^{5/2}_+\\
&\quad - Cs^{-2} Z -CZ^{12/5}s^{1/5} - \mu(N, Z, \uR) N.
\end{split}
\end{equation*}
Optimizing over $s>0$, we get
\begin{align*}
\tr &\left(- \frac{\Delta}{2} -\phi_\mathrm{mol}^\mathrm{TF}\right)\gamma \\
&\ge -2^{5/2}(15 \pi^2)^{-1}\int[\phi_\mathrm{mol}^\mathrm{TF} - \mu(N, Z, \uR)]^{5/2}_+ \\
&\quad - \mu(N, Z, \uR)N - CZ^{25/11} \\
&=-\frac{3}{10}(2/3)(3\pi^2)^{2/3}\int \left[ \left(\rho^\mathrm{TF}_\mathrm{mol}\right)^{5/3}-\mu(N, Z, \uR)\right]^{5/2}_+ \\
&\quad - \mu(N, Z, \uR)N - CZ^{25/11}.
\end{align*}
Using the relation obtained from the TF equation
\begin{align*}
-\mu(N, Z, \uR)N -D\left[\rho_\mathrm{mol}^\mathrm{TF}\right]
&=\frac{3}{10}(5/3)(3\pi^2)^{2/3}\int \left(\rho^\mathrm{TF}_\mathrm{mol}\right)^{5/3} \\
&\quad - \int \rho^\mathrm{TF}_\mathrm{mol} V_{\uR} +D\left[\rho^\mathrm{TF}_\mathrm{mol}\right],
\end{align*}
we learn
\begin{align*}
\tr \left(- \frac{\Delta}{2} -\phi_\mathrm{mol}^\mathrm{TF}\right)\gamma \ge
\E_{\uR}^\mathrm{TF}(\rho_\mathrm{mol}^\mathrm{TF}) + D\left[\rho_\gamma - \rho_\mathrm{mol}^\mathrm{TF}\right] -CZ^{25/11},
\end{align*}
which shows (\ref{lower}).
\end{proof}
We denote 
\[
\Gamma(N, \uZ, \uR) \coloneqq E^\mathrm{TF}_\mathrm{mol}(N, Z, \uR) - \inf \left\{ \sum_{j=1}^K E^\mathrm{TF}_\mathrm{atom}(N_j, Z_j) \colon \sum_{j=1}^K N_j = N \right\}.
\]
It was shown in~\cite[in the proof of Theorem 8]{RuskaiSolovej} that for any pair $(R_i, R_j)$ from $\uR$ there is a decomposition $(N_1, \dots, N_K)$ with $\sum_j N_j = N$ so that
\[
\Gamma(N, \uZ, \uR) \ge \Gamma(N_i + N_j, (Z_i, Z_j), (R_i, R_j)).
\]
From the result in~\cite{BenguriaTF} $\Gamma$ is smallest in the neutral case.
Moreover, it was shown in~\cite{BrezisLiebTF} that $\Gamma(N_i + N_j, (Z_i, Z_j), l(R_i, R_j))l^7$ is an increasing function of $l$ for the neutral case.
 By $|R_i-R_j| > C_0(Z_i + Z_j)^{-1/3}$, with $(R_i, R_j)=  R(Z_i + Z_j)^{-1/3}|R_i-R_j|^{-1}(R_i, R_j)$, we see  $R > C_0$.
 We put $\underline{z_{ij}} \coloneqq (Z_i+Z_j)^{-1}(Z_i, Z_j) $ and $\underline{r_{ij}} \coloneqq |R_i-R_j|^{-1}(R_i, R_j)$ for convenience. Then
 \begin{align*}
 \Gamma(N_i + N_j, (Z_i, Z_j), (R_i, R_j))  
 &\ge (Z_i+Z_j)^{7/3} \Gamma(1, \underline{z_{ij}}  , R\underline{z_{ij}}) \\
 &\ge  |R_i - R_j|^{-7}C_0^7 \Gamma(1, \underline{z_{ij}} , C_0\underline{r_{ij}})\\
 &= C |R_i - R_j|^{-7}.
 \end{align*}
 Here we have used the scaling property of the Thomas-Fermi functional.
 
 Together with these results, we obtain
 \begin{align*}
 \E_{\uR}(\gamma) + U_{\uR} \ge \sum_{i=1}^K E_\mathrm{atom}^\mathrm{TF}(N_i, Z_i)  -CZ^{25/11}  + D\left[\rho_\gamma - \rho_\mathrm{mol}^\mathrm{TF}\right] +CR_\mathrm{min}^{-7}.
 \end{align*}

Next, we show an upper bound for the energy of the M\"uller atom.

\begin{lem}
\label{lem}
For any $N>0$ and $Z>0$
\begin{equation*}
E_\mathrm{atom}(N, Z) \le \E^\mathrm{TF}_\mathrm{atom}(N, Z) + CZ^{11/5}.
\end{equation*}
\end{lem}

\begin{proof}
First, we introduce the reduced Hartree-Fock functional by
\[
 \E_\mathrm{atom}^\mathrm{RHF}(\gamma)
 \coloneqq \tr \left( -\frac{1}{2} \Delta - Z|x|^{-1}\right) \gamma
+D[\rho_{\gamma}].
\]
It is clear that
\[
E_\mathrm{atom}(N, Z) \le \inf\{ \E_\mathrm{atom}^\mathrm{RHF}(\gamma) \colon 0 \le \gamma \le 1, \, \tr \gamma = N\}
\]
We introduce the atomic Thomas-Fermi potential by
\[
\phi_\mathrm{atom}^\mathrm{TF}(x) = Z|x|^{-1} - \rho_\mathrm{atom}^\mathrm{TF} \star |x|^{-1},
\]
where $\rho_\mathrm{atom}$ is the minimizer for the atomic ($K=1$) Thomas-Fermi functional $E^\mathrm{TF}_\mathrm{atom}(N, Z)$ (in the negative ionic situation $N > Z$, we take the neutral TF minimizer).
We apply Lemma \ref{semicl} (ii) with $V=\phi_\mathrm{atom}^\mathrm{TF} - \mu$ ($\mu$ is the chemical potential for the TF atom) and a spherically symmetric $g$ to obtain a density matrix $\gamma'$.
Because of the Thomas-Fermi equation  we see that
\[
\rho_{\gamma'} = 2^{5/2}(6\pi^2)^{-1}(\phi_\mathrm{atom}^\mathrm{TF} - \mu)^{3/2}\star g^2 = \rho_\mathrm{atom}^\mathrm{TF} \star g^2.
\]
Since
\[
\tr \gamma' = \int \rho_{\gamma'} = \int \rho_\mathrm{atom}^\mathrm{TF} = N,
\]
we obtain
\[
\inf\{\E^\mathrm{RHF}(\gamma) \colon 0 \le \gamma \le 1, \, \tr \gamma = N\} \le \E^\mathrm{RHF}(\gamma').
\]
Again, by Lemma \ref{semicl} (ii),
\begin{equation*}
\begin{split}
\E^\mathrm{RHF}(\gamma') &\le 2^{3/2}(5\pi^2)^{-1} \int [V]_+^{5/2} + Cs^{-2}\int[V]_+^{3/2} \\
&\quad  -\int Z|x|^{-1}(\rho_\mathrm{atom}^\mathrm{TF}\star g^2)(x) \, dx +D\left[\rho_\mathrm{atom}^\mathrm{TF}\star g^2\right] \\
&\le 2^{3/2}(5\pi^2)^{-1} \int [V]_+^{5/2} - \int [\phi^\mathrm{TF}_\mathrm{atom} - \mu]\rho_\mathrm{atom}^\mathrm{TF}(x) \, dx  -\mu N \\
&\quad  - D\left[\rho_\mathrm{atom}^\mathrm{TF}\right]
+Z\int(|x|^{-1} - |x|^{-1}\star g^2)\rho_\mathrm{atom}^\mathrm{TF}(x) \, dx  \\
&\quad + Cs^{-2}\int \rho_\mathrm{atom}^\mathrm{TF}   \\
&=-2^{5/2}(15\pi^2)^{-1} \int [\phi^\mathrm{TF}_\mathrm{atom} - \mu]_+^{5/2}
-  D\left[\rho_\mathrm{atom}^\mathrm{TF}\right] 
-\mu N \\
& \quad \quad+ Cs^{-2}\int \rho_\mathrm{atom}^\mathrm{TF} +Z\int(|x|^{-1} - |x|^{-1}\star g^2)\rho_\mathrm{atom}^\mathrm{TF} \\
&=\E_\mathrm{atom}^\mathrm{TF}(\rho_\mathrm{atom}^\mathrm{TF})+ Cs^{-2}\int \rho_\mathrm{atom}^\mathrm{TF} +Z\int(|x|^{-1} - |x|^{-1}\star g^2)\rho_\mathrm{atom}^\mathrm{TF}.
\end{split}
\end{equation*}
In the second inequality, we have used

\begin{equation*}
[g^2 \star |x|^{-1} \star g^2](x-y) \le |x-y|^{-1},
\end{equation*}
as an operator and function.
This is shown, for instance, by using the Fourier transform.
By Newton's theorem, 
\[
0 \le |x|^{-1} - |x|^{-1}\star g^2 = |x|^{-1} \1(|x| \le s).
\]
Then, by the H\"{o}lder inequality,

\begin{equation*}
\begin{split}
Z\int&(|x|^{-1} - |x|^{-1}\star g^2)\rho_\mathrm{atom}^\mathrm{TF} \\
&\le Z\left( \int (\rho_\mathrm{atom}^\mathrm{TF})^{5/3}\right)^{3/5}\left( \int (|x|^{-1} - |x|^{-1}\star g^2)^{5/2}\right)^{2/5} \\
&\le CZ\left( \int (Z|x|^{-1})^{5/2}\right)^{3/5}\left(\int_{|x|\le s} |x|^{-5/2}\right)^{2/5}\, dx \\
&\le CZ^{5/2}s^{1/2},
\end{split}
\end{equation*}
where we have used the Thomas-Fermi equation in the second inequality.
Thus, after optimization in $s$, we arrive at
\[
\E^\mathrm{RHF}(\gamma') \le \E_\mathrm{atom}^\mathrm{TF}(\rho_\mathrm{atom}^\mathrm{TF}) + CZ^{11/5}.
\]
This shows the desired upper bound.
\end{proof}

Inserting this, we obtain
\begin{equation*}
E(N, Z) \ge \sum_{j=1}^K E_\mathrm{atom}(N_j, Z_j) - CZ^{25/11} + D[\rho_\gamma - \rho_\mathrm{mol}^\mathrm{TF}] +CR^{-7}_\mathrm{min}.
\end{equation*}
This completes the proof.\qed

\begin{rem}
It immediately follows that
\begin{equation}
\label{ineq:compTF}
D[\rho_\gamma - \rho_\mathrm{mol}^\mathrm{TF}] \le CZ^{25/11},
\end{equation}
and  $R_\mathrm{min} \ge CZ^{-(1/3)(1- \epsilon)}$ with $\epsilon = 2/77$.
These bounds are the crucial  ingredients when compared with Thomas-Fermi theory.
\end{rem}

\section{Bound on the Positive Excess  Charge}
We assume that a molecule is stable in a configuration $\uR \in \R^{3K}$ and $N < Z$.
Let $\gamma$ be a minimizer for the stable molecule.
The next lemma allows us to localize the M\"uller functional (see~\cite[Lemma 6]{mullerIC}).

\begin{lem}[IMS-type formula]
\label{IMS}
For any quadratic partition of unity $\sum_{j=0}^n \theta_j^2 = 1$ with $\nabla \theta_j \in L^\infty$ and for any density-matrix $\gamma \in \mathcal{P}$, we have
\begin{equation*}
\begin{split}
&\sum_{j=0}^n \E_{\uR}(\theta_j \gamma \theta_j) - \E_{\uR}(\gamma) \\
&\le 
\int_{\R^3} \sum_{j=0}^n |\nabla \theta_j(x)|^2 \rho_\gamma(x) \, dx \\
&\quad + \sum_{i < j}^n \iint_{\R^3 \times \R^3}\frac{\theta_i(x)^2 (|\gamma^{1/2}(x, y)|^2 - \rho_\gamma(x)\rho_\gamma(y))\theta_j(y)^2}{|x-y|} \, dx \, dy.
\end{split}
\end{equation*}
\end{lem}

As in~\cite{RuskaiSolovej} we choose smooth localizing functions $0 \le \theta_j \in C^\infty(\R^3)$, $j = 0, \dots, K$ having the following properties.

\begin{enumerate}
  \setlength{\parskip}{0.1cm}
    \setlength{\itemsep}{0.1cm}
\item[(i)]
For $j \ge 1$ we have $\theta_j(x) = \theta (|x-R_j|/R_\mathrm{min})$, with smooth $\theta $ satisfying $0 \le \theta \le 1$ and $\theta(t) = 1$ if $t < 1/5$ and $\theta(t) = 0$ if $t > 1/4$.

\item[(ii)]
$\sum_{j=0}^K \theta_j(x)^2 = 1$ (which defines $\theta_0$).

 These properties imply
\item[(iii)]
$|\nabla \theta_j(x)| \le CR^{-1}_\mathrm{min}$ for all $j$.
\end{enumerate}

For any $M_1 + M_2 \le M$ we have
\[
E_\mathrm{atom}(M) \le E_\mathrm{atom}(M_1) + E_\infty (M_2).
\]
The proof of this is the same as Proposition~\ref{ifpart} (or, see~\cite[Lemma 2]{Kehle2017}).
Using Proposition~\ref{ifpart}, we have
\begin{equation*}
\begin{split}
\E_{\uR}(\gamma) + U_{\uR}
&\le \sum_{j=1}^K E_\mathrm{atom}(N_j, Z_j) \\
&\le \sum_{j=1}^K\left( E_\mathrm{atom}(N_j^{(1)}, Z_j)  + E_\infty (N_j^{(2)})\right)
\end{split}
\end{equation*}
for any minimizer $\gamma$ and for any $\sum_{j=1}^K (N_j^{(1)} + N_j^{(2)}) = N $.
We note that 
\[
\sum_{j=1}^K E_\infty (N_j^{(2)}) = - \sum_{j=1}^K\frac{N_j^{(2)}}{8}
= -\frac{N^{(2)}}{8}  = E_\infty (N^{(2)})
\]
and take $N_j^{(1)} = \tr (\theta_j \gamma \theta_j)$, $j=1, \dots, K$, and $N^{(2)} = \tr (\theta_0 \gamma \theta_0)$.
Then
\begin{equation}
\label{ineq:binding}
\E_{\uR}(\gamma) + U_{\uR}
\le \sum_{j=1}^K\E_\mathrm{atom}(\theta_j \gamma \theta_j) 
+ \E_\infty(\theta_0 \gamma \theta_0).
\end{equation}
Combining  (\ref{ineq:binding}) with the IMS-type formula in Lemma~\ref{IMS}, we get
\begin{equation}
\label{ineq:IMS}
\begin{split}
0 &\le
\sum_{j=1}^K \E_\mathrm{atom}(\theta_j \gamma \theta_j) +
\E_\infty(\theta_0 \gamma \theta_0) - \E_{\uR}(\gamma) -U_{\uR} \\
&=\sum_{ j=0}^K \E_{\uR}(\theta_j \gamma \theta_j) 
+ \tr (V_{\uR} \theta_0 \gamma \theta_0)
 - \E_{\uR}(\gamma) -U_{\uR} \\
&\quad + \sum_{1 \le i < j\le K} \left(\int_{\R^3} \frac{Z_i \theta_j(x)^2}{|x-R_j|} \rho_\gamma(x) \, dx 
+\int_{\R^3} \frac{Z_j \theta_i(x)^2}{|x-R_i|} \rho_\gamma(x) \, dx \, \right) \\
 &\le 
\int_{\R^3} \sum_{j=0}^K |\nabla \theta_j(x)|^2 \rho_\gamma(x) \, dx
 + \sum_{1\le i < j\le K} I_{ij} + \sum_{j=1}^K I_{0j},
 \end{split}
\end{equation}
where we have denoted
\begin{equation}
\label{error1}
\begin{split}
I_{ij}\coloneqq
&-\frac{Z_iZ_j}{|R_i-R_j|}
+\int_{\R^3} \frac{Z_i \theta_j(x)^2}{|x-R_j|} \rho_\gamma(x) \, dx \\
&\quad 
+\int_{\R^3} \frac{Z_j \theta_i(x)^2}{|x-R_i|} \rho_\gamma(x) \, dx \\
&\quad +\iint_{\R^3 \times \R^3}\frac{\theta_i(x)^2 (|\gamma^{1/2}(x, y)|^2 - \rho_\gamma(x)\rho_\gamma(y))\theta_j(y)^2}{|x-y|} \, dx \, dy
\end{split}
\end{equation}
and
\begin{align*}
I_{0j}&\coloneqq
\int_{\R^3} \frac{Z_j \theta_0(x)^2}{|x-R_j|} \rho_\gamma(x) \, dx \\
&\quad + \iint_{\R^3 \times \R^3}\frac{\theta_0(x)^2 (|\gamma^{1/2}(x, y)|^2 - \rho_\gamma(x)\rho_\gamma(y))\theta_j(y)^2}{|x-y|} \, dx \, dy
\end{align*}
For the first term in (\ref{ineq:IMS}) we learn from the property (iii) of the functions $\theta_j$ that
\begin{equation}
\label{ineq:kineticER}
\int_{\R^3} \sum_{j=0}^K |\nabla \theta_j(x)|^2 \rho_\gamma(x) \, dx
\le CNR^{-2}_\mathrm{min},
\end{equation}
where the constant $C$ depends on $K$.
In order to control the contributions from $I_{ij}$, we now define $N_1^\mathrm{TF}, \dots, N_K^\mathrm{TF}$ to be the positive numbers that minimize $\sum_{j=1}^K E_\mathrm{atom}^\mathrm{TF} (N_j^\mathrm{TF}, Z_j)$ under the constraint $\sum_{j=1}^K N_j^\mathrm{TF} =N$.
Then it is well-known that all the chemical potentials $\mu_\mathrm{atom}(N_j^\mathrm{TF}, Z_j)$ for the atoms will be identical
\[
\mu_\mathrm{atom}(N_j^\mathrm{TF}, Z_j) =\mu_\mathrm{mol}(N, \uZ, \infty), \quad j =1, \dots, K.
\]

\begin{lem}[Lemma 9 in~\cite{RuskaiSolovej}]
\label{lem:TF1}
Let $\rho_\mathrm{mol}^\mathrm{TF}$ be the TF density for the molecular system.
If $C Z^{-1/3} < R' < R_\mathrm{min}/2$ then we have for all $j=1, \dots, K$
\begin{equation}
\label{eq:TF1}
\int_{|x-R_j|<R'} \rho_\mathrm{mol}^\mathrm{TF}(x) \, dx = N_j^\mathrm{TF} + \mathcal{O}(R'^{-3}),
\end{equation}
and if $|x-R_j|>3R_\mathrm{min}/4$
\begin{equation}
\label{eq:TF2}
\int_{|y-R_j|<R'} \rho_\mathrm{mol}^\mathrm{TF}(y)|x-y|^{-1} \, dy = (N_j^\mathrm{TF} + \mathcal{O}(R'^{-3}))|x-R_j|^{-1}.
\end{equation}
\end{lem}

Also we will need the

\begin{lem}[Proposition 10 in~\cite{RuskaiSolovej}]
\label{lem:TF2}
If $\mu_\mathrm{mol}(N, \uZ, \infty)>0$ then there are positive constants $\kappa, \kappa' > 0$ depending on $Z_1, \dots, Z_K$ such that
\begin{equation*}
\kappa < \frac{Z_j - N_j^\mathrm{TF}}{Z_i - N_i^\mathrm{TF}} < \kappa'
\end{equation*}
for all $i \neq j$.
If $\mu_\mathrm{mol}(N, \uZ, \infty)=0$ then $Z_j = N_j^\mathrm{TF}$.
\end{lem}

In order to compare with Thomas-Fermi theory, we use the

\begin{lem}
\label{lem:inside}
Let $\beta > 0$ and $R(Z) = (\beta Z^{-1/3(1- \alpha)})$ with $\alpha < \epsilon = 2/77$ in the previous bound (\ref{ineq:compTF}).
For any fixed $1 \le j \le K$ let $\lambda(x)$ be a function satisfying
\begin{enumerate}
\item[(a)]
$\lambda \in C^{\infty}(\R^3)$ with $0 \le \lambda (x) \le 1$,

\item[(b)]
$\mathrm{supp} \, \lambda \subset \{ x \colon |x-R_j| < R(Z)\}$.

\end{enumerate}

Then there exist $C>0$ and $a > 0$ such that for all small $\alpha < \epsilon$, 

\begin{enumerate}
\item[(i)]

\begin{equation}
\label{eq:inside1}
\left| \int_{\R^3} (\rho_\gamma(x) - \rho_\mathrm{mol}^\mathrm{TF}(x)) \lambda(x) \, dx \right| \le C Z^{(1-a)}.
\end{equation}

\item[(ii)]
If $|y-R_j| > R(Z)$, we have
\begin{equation}
\label{eq:inside2}
\left| \int_{\R^3} \frac{\rho_\gamma(x) - \rho_\mathrm{mol}^\mathrm{TF}(x)}{|x-y|}  \lambda(x) \, dx \right|
\le CZ^{1-a}  |y-R_j|^{-1}.
\end{equation}
\end{enumerate}
\end{lem}

For the proof we need the following lemma for the Coulomb potential (see~\cite[Lemma 18]{TFDWIC}).

\begin{lem}[Coulomb potential estimate]
\label{lem:Coulomb}
For every $f \in L^{5/3}(\R^3) \cap L^{6/5}(\R^3)$ and $x \in \R^3$, we have
\begin{equation*}
\left| \int_{|y| < |x|} \frac{f(y)}{|x-y|} \, dy \right| \le C \| f \|_{L^{5/3}}^{5/6}(|x|D(f))^{1/12}.
\end{equation*}
\end{lem}

\begin{proof}[Proof of Lemma~\ref{lem:inside}]
First, we introduce a function
\[
\Phi_{r}(x)
\coloneqq
 \int_{|y| < r} 
\frac{f(y)}{|x-y|} \, dy.
\]
Applying the Coulomb potential estimate with $f(y) = (\rho_\gamma(y+R_j) - \rho_\mathrm{mol}^\mathrm{TF}(y+R_j)) \lambda(y+R_j)$, we have
\begin{equation*}
\begin{split}
|\Phi_{|x|}(x)|
&=
\left| \int_{|y-R_j| < |x|} 
\frac{\rho_\gamma(y) - \rho_\mathrm{mol}^\mathrm{TF}(y)}{|x-(y-R_j)|} \lambda(y) \, dy \right|\\
&\le
 C \| f \|_{L^{5/3}}^{5/6}(|x|D(f))^{1/12}.
\end{split}
\end{equation*}
By Newton's theorem, we have
\begin{align*}
&\int_{{|y-R_j| < R(Z) } }
 (\rho_\gamma(y) -\rho_\mathrm{mol}^\mathrm{TF}(y))\lambda(y)  \, dy \\
 &\quad = R(Z)  \int_{\S^2} \frac{d \nu}{4\pi} \int_{|y-R_j| < R(Z) }
  \frac{\rho_\gamma(y) - \rho_\mathrm{mol}^\mathrm{TF}(y) }{|R(Z) \nu - (y-R_j)|}\lambda(y) \, dy \\
  &\quad = R(Z)\int_{\S^2} \frac{d \nu}{4\pi}\Phi_{R(Z)}(R(Z) \nu)\\
  &\quad \le CR(Z)^{13/12} \|  \rho_\gamma -\rho_\mathrm{mol}^\mathrm{TF} \|_{L^{5/3}}^{5/6}\left(D\left[\rho_\gamma -\rho_\mathrm{mol}^\mathrm{TF} \right] \right)^{1/12}.
\end{align*}
Combining this with  (\ref{ineq:compTF}) and the kinetic estimates
\[
\int_{\R^3} \rho_\gamma(x)^{5/3} \, dx \le CZ^{7/3}, \quad \int_{\R^3} \rho_\mathrm{mol}^\mathrm{TF}(x)^{5/3} \, dx \le CZ^{7/3},
\]
we find
\begin{align*}
\left|\int_{\R^3}
 (\rho_\gamma(y) -\rho_\mathrm{mol}^\mathrm{TF}(y)) \lambda(y) \, dy \right| 
 \le
 CR(Z)^{13/12} Z^{179/132}.
\end{align*}
Since $179/132 = 49/36 - 1/198$, we have
\[
\left|\int_{\R^3}
 (\rho_\gamma(y) -\rho_\mathrm{mol}^\mathrm{TF}(y)) \lambda(y) \, dy \right| 
 \le  C\beta^{13/12} Z^{1-1/198 + 13\alpha/36}.
\]
Thus if we choose $\alpha < 2/143$, the conclusion (i) follows.

Next, we use the well-known property for subharmonic functions (see~\cite[Lemma 6.5]{TFDWIC}).
\begin{lem}
\label{lem:maximumPR}
Let $f$ be a real-valued function on $\R^3$.
If $f$ is subharmonic for $|x| > r$, continuous for $|x| \ge r$, and vanishing at infinity, then we have
\[
\sup_{|x| \ge r} |x| f(y) = \sup_{|x| = r} |x|f(x).
\]
\end{lem}

We note that $ -\Delta \Phi_r(x) = \1_{|x| < r} (x)f(x)$ and thus harmonic for $|x| > r$.
From  the Coulomb estimate with $r = R(Z)$ and $\pm f(y) = \pm(\rho_\gamma(y+R_j) - \rho_\mathrm{mol}^\mathrm{TF}(y+R_j
) )\lambda(y+R_j)$ we conclude that, on $|y-R_j| > R(Z)$,

\begin{align*}
\left| \int_{\R^3} \frac{\rho_\gamma(x) - \rho_\mathrm{mol}^\mathrm{TF}(x)}{|x-y|} \lambda(x)  \, dx \right|
&\le CZ^{49/36 - 1/198} |y-R_j|^{-1} R(Z)^{13/12} \\
&\le CZ^{1-a}  |y-R_j|^{-1},
\end{align*}
which shows (ii).
\end{proof}

For applying Lemma~\ref{lem:TF1} and Lemma~\ref{lem:inside} we choose $\alpha$ and $\beta$ so that $R_\mathrm{min} > 3 R(Z)$.
If we define $\tilde \theta_j(x) = \theta(|x-R_j|/R(Z))$ for $j \ge 1$ then
\begin{equation}
\label{eq:break1}
\begin{split}
\int_{\R^3} \tilde \theta_j(x)^2 \rho_\gamma(x) \, dx
&= \int_{\R^3} \tilde \theta_j(x)^2 (\rho_\gamma(x) - \rho_\mathrm{mol}^\mathrm{TF}(x) )\, dx  \\
&\quad + \int_{\R^3} \tilde \theta_j(x)^2 \rho_\mathrm{mol}^\mathrm{TF}(x) \, dx \\
&= N_j^\mathrm{TF} + o(Z).
\end{split}
\end{equation}
Thus, since $\sum_{j=1}^K N_j^\mathrm{TF} = N$, we conclude from (\ref{eq:TF1}) and (\ref{eq:inside1}) that
\begin{equation}
\label{eq:density}
\begin{split}
0 &\le \sum_{j=1}^K \int_{\R^3}\rho_\gamma(x) ( \theta_j(x)^2 - \tilde \theta_j(x)^2 )\, dx \\
&\le  \int_{\R^3}\rho_\gamma(x) \left( 1 - \sum_{j=1}^K\tilde \theta_j(x)^2 \right)\, dx 
= o(Z).
\end{split}
\end{equation}
We also get from (\ref{eq:TF2}) and (\ref{eq:inside2})  that
\begin{equation*}
\label{eq:break2}
\begin{split}
\int_{\R^3}\frac{ \tilde \theta_j(x)^2 \rho_\gamma(x) }{|x-R_i|}\, dx
= \frac{N_j^\mathrm{TF} + o(Z)}{|R_i-R_j|}.
\end{split}
\end{equation*}
Using these estimates, we may find

\begin{equation}
\label{eq:break3}
\begin{split}
\int_{\R^3}\frac{ \theta_j(x)^2 \rho_\gamma(x) }{|x-R_i|}\, dx 
&= \int_{\R^3}\frac{\tilde \theta_j(x)^2 \rho_\gamma(x) }{|x-R_i|}\, dx  \\
&\quad + \int_{\R^3}\frac{(\theta_j(x)^2 - \tilde \theta_j(x)^2) \rho_\gamma(x) }{|x-R_i|}\, dx \\
&= \frac{N_j^\mathrm{TF} + o(Z)}{|R_i-R_j|}.
\end{split}
\end{equation}

Next, we estimate the error terms for the direct part of $I_{ij}$.
Combining the above estimates with (\ref{eq:TF2}) in Lemma~\ref{lem:TF1},
\begin{equation*}
\label{ineq:direct}
\begin{split}
&\iint_{\R^3\times \R^3} \frac{\theta_i(x)^2 \theta_j(y)^2 \rho_\gamma(x) \rho_\gamma(y) }{|x-y|} \, dx \, dy \\
&\quad \ge \iint_{\R^3 \times \R^3} \frac{\tilde \theta_i(x)^2 \theta_j(y)^2 \rho_\gamma(x) \rho_\gamma(y) }{|x-y|} \, dx \, dy \\
&\quad \ge
\iint_{\R^3\times \R^3} \frac{\tilde \theta_i(x)^2 \theta_j(y)^2 \rho_\mathrm{mol}^\mathrm{TF}(x) \rho_\gamma(y) }{|x-y|} \, dx \, dy  \\
&\quad \quad  -CZ^{1-a}\int_{\R^3}  \frac{\rho_\gamma(x) \theta_j(x)^2}{|x-R_i|} \, dx \\
&\quad \ge 
(N_i^\mathrm{TF} + o(Z))\int_{\R^3}  \frac{\rho_\gamma(x) \theta_j(x)^2}{|x-R_i|} 
dx.\end{split}
\end{equation*}

Together with (\ref{eq:break3}), we obtain
\begin{align*}
\iint_{\R^3\times \R^3} &\frac{\theta_i(x)^2 \theta_j(y)^2 \rho_\gamma(x) \rho_\gamma(y) }{|x-y|} \, dx \, dy \\
&\ge \frac{ (N_i^\mathrm{TF} + o(Z))   (N_j^\mathrm{TF} + o(Z))}{|R_i-R_j|}.
\end{align*}

For the exchange term in (\ref{error1}), we simply use
\begin{equation*}
\label{error:exchange1}
\begin{split}
\iint_{\R^3 \times \R^3}&\frac{\theta_j(x)^2 (|\gamma^{1/2}(x, y)|^2\theta_i(y)^2}{|x-y|} \, dx \, dy \\
&\le \frac{2}{|R_i-R_j|} \iint_{\R^3 \times \R^3}\theta_j(x)^2 |\gamma^{1/2}(x, y)|^2\, dx \, dy \\
&=  \frac{2}{|R_i-R_j|} \int_{\R^3}\theta_j(x)^2 \rho_\gamma(x)\, dx\\
&\le \frac{N_j^\mathrm{TF} + o(Z)}{|R_i-R_j|}o(Z)
\end{split}
\end{equation*}
by (\ref{eq:break1}) and (\ref{eq:density}).

Thus we arrive at the following estimate for the interaction of two screened nuclei
\begin{equation}
\label{error:conc1}
\begin{split}
I_{ij}
\le
& \frac{-(Z_i - N_i^\mathrm{TF} + o(Z))   (Z_j - N_j^\mathrm{TF} + o(Z))}{|R_i-R_j|}.
\end{split}
\end{equation}

Repeating these arguments, we see
\begin{equation}
\label{error:conc2}
\begin{split}
I_{0j}
\le
\frac{(Z_j - N_j^\mathrm{TF} + o(Z)) o(Z) }{|R_i-R_j|}.
\end{split}
\end{equation}
Inserting the estimates (\ref{ineq:kineticER}), (\ref{error:conc1}) and (\ref{error:conc2}) into (\ref{ineq:IMS}), we get
\begin{equation*}
\begin{split}
0
&\ge
 \sum_{1 \le i < j \le K}\frac{(Z_i - N_i^\mathrm{TF} + o(Z))   (Z_j - N_j^\mathrm{TF} + o(Z))}{|R_i-R_j|} \\
&\quad -CZ^{1+1/3(1-\epsilon)} R_\mathrm{min}^{-1}.
\end{split}
\end{equation*}
If we write $R_\mathrm{min} = |R_{i_0} - R_{j_0}|$ then
\begin{equation*}
\begin{split}
(Z_{i_0} &- N_{i_0}^\mathrm{TF})(Z_{j_0} - N_{j_0}^\mathrm{TF})R^{-1}_\mathrm{min}\\
&\le
 \sum_{1 \le i < j \le K}\frac{(Z_i - N_i^\mathrm{TF}) (Z_j - N_j^\mathrm{TF})}{|R_i-R_j|}  \\
 &\le CZ^{1-\delta}\sum_{ j=1}^K  (Z_j - N_j^\mathrm{TF})R_\mathrm{min}^{-1}
  + CZ^{2(1-\delta)} R_\mathrm{min}^{-1}
\end{split}
\end{equation*}
for some small $\delta > 0$.

If $Z_{i_0} - N_{i_0}^\mathrm{TF} \le CZ^{1-\delta}$, we find from Lemma~\ref{lem:TF2} that $Z_i - N_i^\mathrm{TF} \le CZ^{1-\delta}$ for all $i$.
If $Z_{i_0} - N_{i_0}^\mathrm{TF} \ge CZ^{1-\delta}$, then we divide the above inequality by $Z_{i_0} - N_{i_0}^\mathrm{TF}$ and get $Z_{j_0} - N_{j_0}^\mathrm{TF} \le CZ^{1-\delta}$ because of Lemma~\ref{lem:TF2}.
Again, by Lemma~\ref{lem:TF2}, we see that $Z_i - N_i^\mathrm{TF} \le CZ^{1-\delta}$ for all $i = 1, \dots, K$.
Finally, summing this inequality over $i$, we obtain the desired bound on the positive excess charge
\[
Z-N \le \mathrm{const.} Z^{1-\delta}. 
\]
The proof of the theorem is complete.\qed

\noindent{\bf Acknowlegement:} 
The author would like to thank her supervisor Shu Nakamura for the warm encouragements and helpful comments.
She also thanks Heinz Siedentop for many fruitful discussions and  Phan Th\`anh Nam for introducing ~\cite{RuskaiSolovej} and useful discussions.
This work was supported by Research Fellow
of the JSPS KAKENHI Grant Number 18J13709 and the Program for Frontiers of Mathematical Sciences and Physics, FMSP, University of Tokyo, and the JSPS Bilateral Program: {\it Mathematical modeling of quantum devices affected by phonon}  through ``Open Partnership Joint Research Projects/Seminars".

%% Put all acknowledgements (including those concerning grants) at the end.

\end{document}